\documentclass[letterpaper, 10 pt, conference]{ieeeconf}  
\IEEEoverridecommandlockouts                              
\usepackage{flushend}
\usepackage{graphics} 
\usepackage{stfloats}
\usepackage{epsfig} 
\usepackage{times} 
\usepackage{amssymb}
\usepackage{amsmath} 
\usepackage{url}
\usepackage{indentfirst}
\usepackage{epstopdf}
\usepackage{multirow}
\usepackage{float}
\usepackage{color}
\usepackage{enumerate}
\usepackage{subfigure}
\usepackage{comment}
\usepackage{cite}
\usepackage{caption}
\makeatletter
\newif\if@restonecol
\makeatother

\usepackage[ruled,vlined]{algorithm2e}
\usepackage{algpseudocode}

\newcommand{\ba}{\begin{array}}
	\newcommand{\ea}{\end{array}}
\newcommand{\be}{\begin{equation}}
\newcommand{\ee}{\end{equation}}
\newcommand{\bea}{\begin{eqnarray}}
\newcommand{\eea}{\end{eqnarray}}
\newcommand{\bean}{\begin{eqnarray*}}
	\newcommand{\eean}{\end{eqnarray*}}
\newcommand{\bc}{\begin{center}}
	\newcommand{\ec}{\end{center}}

\title{\LARGE \bf
 Fault Diagnosis of Discrete-Event Systems under\\  Non-Deterministic Observations with Output Fairness
 }

\author{Weijie Dong, Shang Gao, Xiang Yin and Shaoyuan Li
	\thanks{This work was supported by the National Natural Science Foundation of China (61803259, 61833012).}
	\thanks{W. Dong S. Gao,  X. Yin and S. Li  are with Department of Automation and Key Laboratory of System Control and Information Processing,
		Shanghai Jiao Tong University, Shanghai 200240, China.
		E-mail: {\tt\small  \{wjd\_dollar,gaoshang,yinxiang,syli\}@sjtu.edu.cn}.}
}
\newtheorem{mydef}{Definition}
\newtheorem{mythm}{Theorem}

\newtheorem{mylem}{Lemma}

\newtheorem{myexm}{Example}
\newtheorem{remark}{Remark}

\IncMargin{0.8em}
\begin{document}

	\maketitle
	
\begin{abstract}
In this paper, we revisit the fault diagnosis problem of discrete-event systems (DES) under non-deterministic observations. 
Non-deterministic observation is a general observation model that includes the case of  intermittent loss of observations.
In this setting, upon the occurrence of an event, the sensor reading may be non-deterministic such that a set of output symbols are all possible. 
Existing works on fault diagnosis under non-deterministic observations require to consider  all possible observation realizations. 
However, this approach includes the case where some possible outputs are permanently disabled. 
In this work,  we introduce the concept of \emph{output fairness} by requiring that, for any output symbols, 
if it has infinite chances to be generated, then it will indeed be generated infinite number of times. 
We use an assume-guarantee type of linear temporal logic formulas to formally describe this assumption. 
A new notion called output-fair diagnosability (OF-diagnosability) is proposed. 
An effective approach is provided for the verification of OF-diagnosability. 
We show that the proposed notion of OF-diagnosability is weaker than the standard definition of diagnosability under non-deterministic observations, and 
it better captures the physical scenario of observation non-determinism or intermittent loss of observations. 
\end{abstract}
	
	\section{Introduction}
	Engineering cyber-physical systems (CPS), such as manufacturing systems, transportation systems and power systems, are generally very complex due to their intricate operation logic and hybrid dynamics.  
	For such large-scale safety-critical systems, failures during their operations are very common as millions of components/modules are working in parallel. 
	Therefore, failure diagnosis and detection is crucial but challenging task in order to monitor the operation conditions and to maintain safety for CPSs. 
 
	In this work, we investigate the fault diagnosis problem in the framework of discrete-event systems (DES), which are widely used in modeling the high-level logical behaviors of CPSs \cite{cassandras2008introduction}. 
	In the context of DES, the problem of fault diagnosis was initiated by \cite{sampath1995diagnosability},
	where the notion of \emph{diagnosability} was proposed. Specifically, it is assumed that the system's events are partitioned as observable events and unobservable events, 
	and
	a system is said  to be diagnosable if the occurrence of fault event can always be detected within a finite delay based on the observed sequence. To the past years, fault diagnosis of DES remains a hot topic due to its importance; see, e.g., some recent works \cite{lin2017n,yin2017decidability,ran2018codiagnosability, yin2019robust,viana2019codiagnosability,hu2021diagnosability}.
	The reader is referred to the comprehensive survey papers \cite{zaytoon2013overview,lafortune2018history} for more details on this topic. 

    In the modeling of DES, the occurrences of events are essentially observed by the corresponding sensors. 
    In practice, however, due to measurement noises, sensor failures or malicious attacks, 
    the sensor readings can be \emph{unreliable} or \emph{non-deterministic}.  
    Such non-deterministic observation issue was initially addressed in the context of \emph{intermittent loss of observations}  \cite{carvalho2012robust,lin2014control,boussif2021intermittent}, where it is assumed that some observable events are unreliable in the sense that their occurrences may be observed or be lost non-deterministically. 
    In \cite{takai2012verification}, Takai and Ushio investigated the issue of non-deterministic observation using Mealy automata. Specifically, the observation of the system is modeled by a \emph{state-dependent  non-deterministic output function}. 
    The model is quite general that captures the issue of intermittent loss of observations. 
    Furthermore, it allows the observation symbols to be different from the original event set. 
 
	Our work is motivated by the aforementioned works on fault diagnosis under intermittent loss of observations \cite{carvalho2012robust} or, in a broad sense, non-deterministic observations \cite{takai2012verification}. 
	In particular, we note that the existing   models  for non-deterministic observations essentially assume  that \emph{all possible  realizations under the non-deterministic output mapping are feasible}.  
	In the context of fault diagnosis, therefore, it needs to consider whether or not fault can be detected under all possible observation realizations. 
	We argue in this work that this setting is somehow too strict for the purpose of diagnosis since it may exclude the possibility for the occurrence of some possible sensor reading, which is ``unfair". 
	
	To see this ``unfairness" more clearly, let us consider the case of intermittent loss of observations. 
	Suppose that after the occurrence of fault events, there will be an observable but unreliable event occurring repeatably.
	Furthermore, once this indicator event is observed, we  know for sure that the fault has occurred. 
	In the   existing frameworks \cite{takai2012verification,carvalho2012robust}, this system is not diagnosable because we should consider the case where the observation of this unreliable event \emph{is lost each time when it occurs}. However, this scenario actually corresponds to the case that the underlying sensor has failed \emph{permanently}. Of course, such a permanent sensor failure is also possible in practice \cite{kanagawa2015diagnosability,takai2021general,carvalho2013robust,carvalho2021comparative}, 
	but it should be categorized differently from the setting of  intermittent loss of observations. 
	In other words, for such a sensor that is unreliable \emph{but we have prior information that it will not fail permanently}, the fault may be diagnosed with unreliability.
	To this end, Vinicius S. et al in \cite{oliveira2020k} proposed K-diagnosability which excludes permanent sensor failures by bounding the maximum number of consecutive failure for an unreliable sensor under a finite integer $k$. 
	However, in real-word engineering, we usually only can be sure that a sensor will not fail permanently and can repair after finite time delay, but we have no idea about what the upper bound of time needed to repair is. Thus, in this paper, we use the notion of "fair" to solve this problem.

	It is worth noting that, in \cite{biswal2015polynomial}, the authors investigated diagnosability for \emph{fair systems}. 
	However, the notion of fairness is different from our setting. 
	Specifically, \cite{biswal2015polynomial} assumes that the dynamic of the system is fair in the sense that 
	each transition can be executed for infinite number of times whenever it is enabled infinitely. 
	However, the observation mapping considered therein is still static modeled as a natural projection.
	In contrast, we  impose fairness constraint on the observation mapping rather than the internal behavior of the system 
	the system's dynamic.  Therefore, the problem setting in our work is quite different from that of \cite{biswal2015polynomial}.

	In this paper, we  revisit the fault diagnosis problem under non-deterministic observations. 
	Motivated by the above discussion, however, to better capture the physical scenario in which \emph{all possible observations are always available}, 
	we further require that each possible observation (or output symbol) is \emph{fair} in the sense that if it has infinite chances to be observed then it will indeed be observed infinite number of times. 
    To formally describe this fairness requirement, we use an assume-guarantee type linear temporal logic (LTL) formulas. 
	Then, we propose a new condition called \emph{output-fair diagnosability} (OF-diagnosability) that provides the necessary and sufficient condition for the existence of a diagnoser that works correctly under the fair-observation setting. 
	Also, we utilize the structure properties of fairness requirement to provide an more effective procedure for checking this new condition. 
	Our work bridges the gap between the existing mathematical definition of diagnosability under intermittent loss of observations and the physical setting, where those unreliable sensors will not fail permanently. 

	\section{Preliminary}\label{sec:2}
	\subsection{System Model}
	Let $\Sigma$ be a finite set of events. 
	A finite (infinite) string over $\Sigma$  is  a finite (infinite) sequence of events in $\Sigma$ of form $s=\sigma_1\sigma_2 \dots \sigma_n (\dots)$,  where $\sigma_i\in \Sigma$, for $i = 1,\cdots , n$. 
	A finite (infinite) language is a set of finite (infinite) strings. 
	We denote by $\Sigma^*$ and $\Sigma^\omega$, respectively, the set of all finite and the set of all infinite strings over $\Sigma$. 
	Note that the empty string, denoted by $\varepsilon$, is included in $\Sigma^*$. 
	Given a finite language $L \subseteq \Sigma^*$, the prefix-closure of $L$ is defined by 
	$\textsf{Pre}(L)=\{w \in \Sigma^*: \exists t \in \Sigma^* \text{ s.t. }  w t\in L\}$. 
	Similarly, for an infinite language $L\subseteq  \Sigma^\omega$,  
	its prefix-closure is the set of all its finite prefixes, i.e., 
	$\textsf{Pre}(L)=\{w \in \Sigma^*:\exists t \in \Sigma^\omega \text{ s.t. } w t \in L\}$. 
	For any string $s\in \Sigma^*\cup \Sigma^\omega$, we write $\textsf{Pre}(\{s\})$ as $\textsf{Pre}(s)$ for the sake of simplicity. 
	Given an infinite string $s \in \Sigma^\omega$,  $\textsf{Inf}(s)$ denotes the set of events that appear infinitely number of times in the string.
	
In this paper, we consider DES modeled by a deterministic finite-state automaton (DFA) 
$G = (Q,\Sigma,f,q_0)$, 
where  $Q$ is the finite set of states;
	 $\Sigma$ is the finite set of events;
	  $f:Q\times \Sigma \to Q$ is the partial transition function; and 
	   $q_0\in Q$ is the initial state.  
The transition  function $f$ can be extended to $f:Q\times \Sigma^* \to Q$ recursively in the usual manner; see, e.g.,  \cite{cassandras2008introduction}. 
The finite language generated by $G$ is defined by $\mathcal{L}(G)=\{s\in \Sigma^*:f(q_0,s)!\}$, where ``!" means ``is defined", 
and the infinite language generated by $G$ is defined by $\mathcal{L}^\omega(G)=\{s\in \Sigma^\omega:f(q_0,s)!\}$. 
We assume that system $G$ is live, i.e., $\forall q\in Q,\exists \sigma: f(q,\sigma)!$. 

\subsection{Non-Deterministic Observations}
In the partial observation setting, it is assumed that the occurrence of each event cannot be observed directly or perfectly. 
A commonly adopted simple approach  is to partition the event set $\Sigma$ into   observable events $\Sigma_o$ and unobservable events $\Sigma_{uo}$. 
Then natural projection $P:\Sigma^*\to \Sigma_{o}^*$ can be used to capture the issue of partial observability.

In many real-world scenarios, however, the sensor readings may   be non-deterministic due to observation noises, sensor failures or malicious attacks. 
Furthermore, the sensor reading of the  event  occurrence can be \emph{state-dependent}. 
To this end, more complicated observation models have been proposed in the literature. 
Here, we adopt state-dependent non-deterministic observation model proposed by \cite{takai2012verification,ushio2015nonblocking}, 
which captures both  state-dependency and  non-determinism of observations.  

Formally, we assume that $\Delta$ is a new set of all possible observations or \emph{output   symbols}. 
Then a state-dependent non-deterministic output function is 
\[
\mathcal{O}: Q\times \Sigma \to 2^{\Delta_{\varepsilon}}, 
\]
where $\Delta_{\varepsilon}=\Delta \cup \{\varepsilon\}$. 
Intuitively, this output function means that 
if event $\sigma\in\Sigma$ occurs at state $q\in Q$, 
then the system is possible to observe any symbol in $\mathcal{O}(q,\sigma)$ non-deterministically. 

\begin{remark} \label{rem:capture intermittent}
The above non-deterministic observation model is quite general in the sense it subsumes many observation models in the literature. 
For example, the standard natural-projection-based observation   can be formulated by 
setting $\mathcal{O}(q,\sigma)=\{\sigma\}$ for all $\sigma\in \Sigma_o$
and $\mathcal{O}(q,\sigma)=\{\varepsilon\}$ for all $\sigma\in \Sigma_{uo}$. 
Also, it captures the so-called \emph{intermittent loss of observations} \cite{carvalho2012robust}. 
In this setting,  the event set is usually partitioned as 
$\Sigma=\Sigma_{r} \dot{\cup} \Sigma_{ur} \dot{\cup} \Sigma_{uo}$, 
where $\Sigma_{r}$ is the set of reliable events whose occurrences can  always be observed directly, 
    $\Sigma_{ur}$ is the set of unreliable events whose occurrences may be observed but can also be lost and 
    $\Sigma_{uo}$ is the set of unobservable events whose occurrences can  never be observed. 
This setting can be captured by considering  a  non-deterministic output function with  $\Delta=\Sigma$ and 
for any $q \in Q$ and  $\sigma \in \Sigma$, we have
	\[
		\mathcal{O}(q,\sigma)\!=\!
		\left\{\!\!\!\! 
		\begin{array}{l l}
			\{\sigma\} &\text{if } \sigma \!\in\! \Sigma_r \\
			\{\sigma,\varepsilon\} &\text{if } \sigma \!\in\! \Sigma_{ur} \\ 
			\{\varepsilon\} &\text{if } \sigma \!\in\! \Sigma_{uo} 
		\end{array} 
		\right.
	\]
Note that, for the   general case  we consider here,  the output symbols $\Delta$ can be different from the original event set $\Sigma$. 
\end{remark}

Based on the output function $\mathcal{O}: Q\times \Sigma \to 2^{\Delta_{\varepsilon}}$, 
we can define a non-deterministic observation mapping $M:\mathcal{L}(G) \to 2^{\Delta^*}$, where $\Delta^*$ is the set of finite strings over $\Delta$ and we have $\varepsilon \in \Delta^*$, recursively as:
\begin{itemize}
	\item $M(\varepsilon)=\{\varepsilon\}$;
	\item 
	for any $s\in \Sigma^*$ and $\sigma\in \Sigma$, we have 
	\[
	M(s \sigma)
	=\{\alpha \beta \!\in\! \Delta^*: \alpha \!\in\! M(s) \wedge \beta\in \mathcal{O}(f(q_0,s),\sigma) \}
	\]
\end{itemize}
Intuitively, $M(s)\in \Delta^*$ is the set of all possible observations (or output strings) upon the occurrence of internal string $s \in \Sigma^*$. 
The  observation mapping is also extended to $M: 2^{\Sigma^*} \to 2^{\Delta^*}$ by: for any language  $L \subseteq \mathcal{L}(G)$, $M(L)=\cup_{ s \in L}  M(s)$. 

Since the observation is non-deterministic, 
for any specific internal string $s\in \mathcal{L}(G)$, it may have different \emph{output realizations}. 
To ``embed" the actual observation occurred into the internal execution of the system, 
we define 
\[
\Sigma_e=Q\times \Sigma\times \Delta_\varepsilon
\]
as the set of \emph{extended events}. 
Then an extended string is a finite or infinite sequence of extended events. 
We say a finite extended string
\[
s=(q_0,\sigma_0,\delta_0)(q_1,\sigma_1,\delta_1)\dots (q_n,\sigma_n,\delta_n) \in \Sigma_e^*
\]
is generated by $G$ if 
$f(q_i,\sigma_i)=q_{i+1},\forall i<n$ and $\delta_i\in \mathcal{O}(q_i,\sigma_i),\forall i\leq n$. 
We denote by $\mathcal{L}_e(G)$ the set of all finite extended strings generated by $G$. 
The infinite extended strings are defined analogously and the set of all infinite extended strings generated by $G$
is denoted by $\mathcal{L}^{\omega}_e(G)$. 

For any extended string
$s=(q_0,\sigma_0,\delta_0) (q_1,\sigma_1,\delta_1)\cdots$, 
we define $\Theta_{Q}(s)=q_0 q_1\cdots, \Theta_{\Sigma}(s)=\sigma_0\sigma_1 \cdots$ and $\Theta_{\Delta}(s)=\delta_0\delta_1\cdots$ as its corresponding state sequence, (internal) event string and output string, respectively.
Clearly, for any  $s \in \mathcal{L}_e(G)$, we have $\Theta_{\Delta}(s) \in M(\Theta_{\Sigma}(s))$
because $\Theta_{\Delta}(s)$ is a specific  observation realization of internal string $\Theta_{\Sigma}(s)$.  
\begin{figure} 
	\centering
	\centering
	\includegraphics[width=3.5cm]{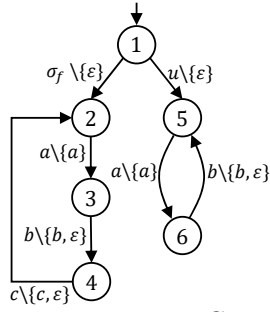}
	\caption{\label{fig:G-1}System $G_1$.}
\end{figure}
\begin{myexm} \label{exam:system example}
	Let us consider system $G_1$ in Figure~\ref{fig:G-1}, where we have $\Sigma=\{a,b,c,\sigma_{\!f},u\}$ and $\Delta=\{a,b,c\}$. 
	The output function $\mathcal{O}: Q\times \Sigma \!\to\! 2^{\Delta_{\varepsilon}}$ is specified by the label of each transition, 
	where the LHS of   ``$\backslash $" denotes the internal event 
	and  the RHS of   ``$\backslash $" denotes the set of all possible output symbols. 
	For example, $b \backslash  \{b,\varepsilon\}$ from state $3$ to state $4$ means that system moves to state $4$ from state $3$ by firing event $b$, i.e., $f(3,b)=4$, and we can non-deterministically observe any output in set $\mathcal{O}(3,b)= \{b,\varepsilon\}$. 
	In this example, we either observe $b$ itself or observe nothing, which corresponds to the case of intermittent loss of observations. 
	Then for finite string  $\sigma_{\!f} a b \in \mathcal{L}(G_1)$, 
	the set of all possible output is $M(\sigma_{\!f} a b)=\{ab,a\}$.
	These two different outputs lead to the following two extended strings   $(1,\sigma_{\!f},\varepsilon)(2,a,a)(3,b,b)$ and  $(1,\sigma_{\!f},\varepsilon)(2,a,a)(3,b,\varepsilon)$, respectively.
\end{myexm}


\subsection{Fault Diagnosis}
In this work, we assume that the system is subject to faults whose occurrences need to be diagnosed within a finite number of steps. 
Specifically,  we assume that $\Sigma_F \subset \Sigma$ is the set of fault events. 
For the sake of simplicity, we do not distinguish among different fault types. 
We say a string $s\in \Sigma^*\cup\Sigma^\omega$ is faulty if some event in $\Sigma_F$ appears in $s$
and we write $\Sigma_F\in s$ with a slight abuse of notation; otherwise, it is normal. 
We define $\mathcal{L}_F(G)$ and $\mathcal{L}_F^\omega(G)$ as the sets of all finite and infinite faulty strings generated by $G$, respectively. 
Similarly, we define $\Sigma_{e,F}= Q\times \Sigma_F\times \Delta^\varepsilon$ as the set of extended fault events 
and denote by  $\mathcal{L}_{e,F}(G)=\{s\in \mathcal{L}_{e}(G): \Sigma_{e,F}\in s\}$ the set of all finite extended faulty strings generated by $G$; the extended infinite faulty language $\mathcal{L}_{e,F}^\omega(G)$ is also defined analogously. 
Finally, we define $\Psi_e(G)$ as the set of all finite extended faulty strings in which fault events occur \emph{for the first time}, i.e., 
\begin{align}
	\Psi_e(G)=\{s \in \mathcal{L}_{e,F}(G): \forall t \in \textsf{Pre}(s) \setminus \{s\},  \Sigma_{e,F} \notin  t \}. \nonumber
\end{align} 
To capture whether or not the occurrences of fault events can be detected within a finite number of steps, 
the notion of diagnosability (under non-deterministic observations) has been proposed in the literature \cite{takai2012verification}. 
\begin{mydef} \label{def:classical diag}
System $G$ is said to be \emph{diagnosable} w.r.t.\  output function $\mathcal{O}$ and fault events $\Sigma_F$ if 
\begin{equation} \label{eq:classical diagnosis}
  (\forall s \!\in\! \Psi_e(G) ) (\exists n\!\in\! \mathbb{N}) (\forall st  \!\in\! \mathcal{L}_e(G)) 
  [ |t| \geq n \Rightarrow \textsf{diag} ]
\end{equation}
where the diagnostic condition $\textsf{diag}$ is 
\[
  (\forall \omega \in \mathcal{L}_e(G))  
  [\Theta_{\Delta}(\omega)=\Theta_{\Delta}(st) \Rightarrow    \Sigma_{e,F}\in \omega]. \vspace{8pt}
\] 
\end{mydef}

Intuitively, the above definition says that,  
for any faulty extended string in which fault events appear for the first time, there exists a finite detection bound such that,   
for any of its continuation longer than the detection bound, 
any other extended strings having the same observation must also contain fault events. Note that we consider extended strings rather than the internal strings in order to capture the issue of non-deterministic observations.

\begin{myexm} \label{exam:classical not diagnosable}
Again, we consider system $G_1$ depicted in Figure~\ref{fig:G-1} with $\Sigma=\{a,b,c,\sigma_{\!f},u\} , \Delta=\{a,b,c\}$ and $\Sigma_F=\{\sigma_{\!f}\}$. 
Let us consider faulty extended string  $(1,\sigma_{\!f},\varepsilon) \in \Psi_e(G_1)$, which can be extended   arbitrarily long as
\[
s_F=(1,\sigma_{\!f}, \varepsilon)[(2,a,a)(3,b,b)(4,c,\varepsilon)]^n
\]
However, for any $n$, we can find a normal extended string 
\[
s_N=(1,u, \varepsilon)[(5,a,a)(6,b,b)]^n
\]
such that $\Theta_{\Delta}(s_F)=\Theta_{\Delta}(s_N)=(ab)^n$. 
Therefore, we know that system $G_1$ is not diagnosable under $\mathcal{O}$.
\end{myexm}

\section{Diagnosability with Always-Fairness Assumption}\label{sec:3}

\subsection{Motivating Example}
As we discussed in Remark~\ref{rem:capture intermittent}, the non-deterministic output function can be used to capture the issue of \emph{intermittent loss of observations}. 
Here, however,  we argue that the original definition of diagnosability in Definition~\ref{def:classical diag} may be too strong for the case of intermittent loss of observations. 

To see this, we still consider system $G_1$ shown in Figure~\ref{fig:G-1}, where the reliable event set is $\{a\}$, the unreliable events set is $\{b,c\}$ and the unobservable event set is $\{\sigma_{\!f},u\}$. 
As we have discussed in Example~\ref{exam:classical not diagnosable}, this system is not diagnosable because 
there are two infinite extended strings:
\begin{itemize}
    \item 
    one is faulty $s_F=(1,\sigma_{\!f}, \varepsilon)[(2,a,a)(3,b,b)(4,c,\varepsilon)]^\omega$; 
    \item 
    the other is non-faulty $s_N =(1,u, \varepsilon)[(5,a,a)(6,b,b)]^\omega$
\end{itemize}
and they have the same observation, i.e., $\Theta_{\Delta}(s_1)=\Theta_{\Delta}(s_2)$ $=(ab)^\omega$. 
As a result, the occurrence of fault in string $s_F$ can never be detected within a finite number of steps. 

Note that, the existence of above extended string $s_F$ is due to the following physical scenario: 
the infinite internal string $\sigma_{\!f}(abc)^\omega$ is executed, i.e., the system loops forever in the cycle formed by states $2,3$ and $4$, 
and each time when event $c$ is executed at state $4$, the sensor reads $\varepsilon$ due to observation loss.  
In other words, the sensor corresponds to transition $4\xrightarrow{c}$ has to  fail \emph{permanently} in order to draw the conclusion that the system is not diagnosable.  
However, this source of non-diagnosability seems violate the setting of intermittent loss of observations. 
In the context of intermittent loss of observations or non-deterministic observations, 
it makes more sense to assume that each possible observation is \emph{eventually possible}. 
Therefore, if event $c$ at state $4$ corresponds to a sensor that may be unreliable but will not fail permanently, 
then along the infinite loop of $\sigma_{\!f}(abc)^\omega$, once one has a single chance to observe output $c$ for transition $4\xrightarrow{c}$, 
it will conclude immediately that fault events have occurred.

\subsection{$\Delta'$-Fairness Assumption}
The above discussion suggests that the exiting definition of diagnosability under non-deterministic observations is a bit strong than its underlying physical setting because it includes the situation where some output symbols are disabled permanently. 
In order to bridge this discrepancy between the definition of diagnosability and the physical meaning of non-deterministic observations, in this work, 
we put an additional \emph{fairness} assumption on those sensors that are subject to intermittent loss of observations, or non-deterministic observations in a broad sense. 
Specifically, we assume that $\Delta'\subseteq \Delta$ is a set of ``fair" output symbols 
in the sense that if they have infinite opportunities to occur, then they will indeed occur infinite number of times.  

In order to formalize this fairness requirement, we use the linear temporal logic (LTL) formulas. 
Formally,  an LTL formula 
$\varphi$ is constructed based on a set of atomic propositions $\mathcal{AP}$, Boolean operators and temporal operators as follows:
\[
\varphi::=true\ |\ p\ |\ \varphi_1 \wedge \varphi_2\ |\ \neg \varphi\ |\  \bigcirc \varphi \ |\  \varphi_1 U \varphi_2 ,
\]
where $p\in \mathcal{AP}$ is an atomic proposition; $\bigcirc$ and $U$ denote ``next" and ``until", respectively. Other Boolean connectives can be induced by $\wedge$ and $\neg$, e.g., $\varphi_1 \vee \varphi_2 = \neg (\neg \varphi_1 \wedge \neg \varphi_2)$ and $\varphi_1 \rightarrow \varphi_2 = \neg \varphi_1 \vee \varphi_2$. 
Temporal operators $\Box$  ``always"  and $\lozenge$  ``eventually"  can be induced by until operator, i.e., $\lozenge \varphi = true U \varphi$ and $ \Box \varphi = \neg \lozenge \neg \varphi$. 
 LTL  formulas are evaluated over infinite sequences of atomic proposition sets (called words). 
For any infinite word  $s \in (2^\mathcal{AP})^\omega$, 
we denote by $s \models \varphi$ if it satisfies LTL  formula $\varphi$. 
The reader is referred to \cite{baier2008principles}  for more details on the semantics of LTL. 

In our context of non-deterministic observations, we choose extended events as atomic propositions, 
i.e., $\mathcal{AP}=\Sigma_e$.  
For the sake of simplicity, each extended event $(q,\sigma,\delta)$ will also be written as $\sigma_q^\delta$ meaning that 
event $\sigma$ occurs at state $q$ and generates observation $\delta$. 
For simplicity, 
we define 
\[
\sigma_q:=  \bigvee_{\delta\in \mathcal{O}(\sigma,q)} \sigma_q^\delta
\]
as the proposition formula meaning that transition $q\xrightarrow{ \sigma}$ occurs  no matter what output it generates.
Then we introduce the notion of $\Delta'$-fairness as follows. 

\begin{mydef} \label{def:OF}  	($\Delta'$-Fairness) 
Let $\Delta'\subseteq \Delta_\varepsilon$ be a set of output symbols. 
We say an infinite extended string $s \in \mathcal{L}_e^\omega(G)$ is   $\Delta'$-fair if 
\begin{equation}\label{eq:LTL-form}
s\models \varphi_{fair}\!:=\!
\bigwedge_{q\in Q, \sigma\in \Sigma} 
\left( \Box \lozenge \sigma_q  \to    \bigwedge_{\delta\in \Delta' \cap \mathcal{O}(q,\sigma)}  \Box \lozenge \sigma_q^\delta     \right)
\vspace{8pt}
\end{equation}
\end{mydef} 

The above subset of outputs $\Delta'\subseteq \Delta_\varepsilon$ is referred to as the \emph{fair outputs} meaning that these outputs will always have the opportunity to be measured. 
Specifically, this requirement is captured by   formula $\varphi_{fair}$   saying that, 
for any transition $q\xrightarrow{ \sigma}$, if it is fired infinite number of times, 
then any of its fair outputs, i.e., $\sigma\in \Delta'\cap \mathcal{O}(q,\sigma)$, can actually be observed infinite number of times. 
This excludes the case where some fair outputs are disabled permanently. 
We define 
\[
\mathcal{L}_{e}^{\varphi}(G)=\{s \in \mathcal{L}^{\omega}_{e}(G): s \models  \varphi_{fair} \}
\]
as the set of infinite extended strings generated by $G$ satisfying $\varphi_{fair}$. 
We also define $\mathcal{L}^{\varphi}_{e,F}(G)=\mathcal{L}^{\varphi}_{e}(G) \cap \mathcal{L}_{e,F}^\omega(G)$
as the set of infinite faulty extended strings satisfying $\varphi_{fair}$.

\begin{remark}
Throughout the paper, we will only consider LTL formulae in the form of Equation~\eqref{eq:LTL-form} to capture the fairness requirement.  
Compared with the general form of LTL,  Equation~\eqref{eq:LTL-form} is an assume-guarantee type of formula. Later,  we will utilize this  structural property for the purpose of checking diagnosability. 
\end{remark}

\subsection{OF-Diagnosability}
Based on the above notion of $\Delta'$-fairness on  extended strings, now we modify the existing definition of diagnosability as shown in Definition~\ref{def:classical diag} to the \emph{output-fair diagnosability} (OF-diagnosability) defined as follows. 
\begin{mydef}\label{def:OF-diagnosability}
(OF-diagnosability) 
System $G$ is said to be \emph{output-fairly diagnosable} (OF-diagnosable) w.r.t.\  fault events $\Sigma_F$, output function $\mathcal{O}$ and fair outputs $\Delta'\subseteq \Delta_\varepsilon$ if 
\begin{equation} \label{eq:of-diagnosis}
 (\forall s \in \mathcal{L}^{\varphi}_{e,F}(G))(\exists t \in \textsf{Pre}(s))[ \textsf{fair-diag} ]  
\end{equation}
where the fair-diagnostic condition $\textsf{fair-diag}$ is 
\[
  (\forall w \in   \mathcal{L}_e(G)   )  
  [\Theta_{\Delta}(w)=\Theta_{\Delta}(t) \Rightarrow    \Sigma_{e,F}\in w]. \vspace{8pt}
\]   
\end{mydef}\vspace{4pt}

Intuitively, OF-diagnosability revises the standard diagnosability by restricting our attention only to those infinite extended strings satisfying the fairness assumption 
and investigates whether or not faults in those ``output-fair" strings can be detected. 

\begin{remark}
In Definition \ref{def:classical diag}, it is known that 
``$\forall s \!\in\! \Psi_e(G)$" and  ``$\exists n\!\in\! \mathbb{N}$"
can be swapped, which means that if the system is diagnosable, then there exists a uniform detection bound after the occurrence of any fault events. 
However, in Definition \ref{def:OF-diagnosability}, 
the length of detection prefix $t$ can be arbitrarily long, 
depending on how the fairness assumption is satisfied in the specific infinite faulty string $s \in \mathcal{L}^{\varphi}_{e,F}(G)$, 
since the  $\Delta'$-fairness assumption only guarantees that   all fair outputs \emph{eventually} occur but the delay can be arbitrarily large.
\end{remark}

We show that the proposed notion of a OF-diagnosability provides the necessary and sufficient condition for the existence of diagnoser that works ``correctly" under the $\Delta'$-fairness assumption. 
Formally, a diagnoser is a function
\[
D: M(\mathcal{L}(G)) \to \{0,1\}
\]
that decides whether a fault has happened (by issuing ``$1$") or not (by issuing ``$0$") based on the output string. 
We say that a diagnoser works correctly under the $\Delta'$-fairness assumption if it satisfies the following conditions:
\begin{enumerate}[C1)]
	\item 
	By assuming that each fair-output will actually be observed infinitely if they have infinite chances to be observed, the diagnoser will issue a fault alarm for any occurrence of fault events, i.e., 
	\[
	(\forall s  \in \mathcal{L}^{\varphi}_{e,F}(G))(\exists s' \in \textsf{Pre}(s)))[D(\Theta_{\Delta}(s'))=1].\label{C1}
	\]
	\item 
	The diagnoser will not issue a false alarm if the execution is still normal, i.e., 
	\[
	(\forall s \in \mathcal{L}_{e}(G): \Sigma_{e,F} \notin s) [D(\Theta_{\Delta}(s)) = 0].
	\label{C2}
	\]
\end{enumerate}

The following theorem says that there exists a diagnoser working ``correctly" under the $\Delta'$-fairness assumption if and only if the system is OF-diagnosable.
\begin{mythm}\label{thm:OF-existence}
There exists a diagnoser   satisfying conditions C\ref{C1} and C\ref{C2} 
if and only if $G$ is OF-diagnosable w.r.t.\  fault events $\Sigma_F$, output function $\mathcal{O}$ and fair outputs $\Delta'\subseteq \Delta_\varepsilon$.
\end{mythm}
\begin{proof}
($\Rightarrow$) 
Suppose that there exists a diagnoser $D$ satisfying conditions C1 and C2, while, by contradiction, system $G$ is not OF-diagnosability. 
This means that there exist an infinite faulty extended string 
$s \in \mathcal{L}_{e,F}^{\varphi}(G)$ 
such that for any prefix $t\in \textsf{Pre}(s)$, there exists a normal extended string 
$w \in \mathcal{L}_e(G)$ having the same observation with $t$. 
Since diagnoser $D$ satisfies condition C2, 
for any $t \in \textsf{Pre}(s)$, we get $D(\Theta_{\Delta}(t))=0$; otherwise, if $D(\Theta_{\Delta}(t))=1$, we get $D(\Theta_{\Delta}(\omega))=D(\Theta_{\Delta}(t))=1$, which violates condition C2.
Therefore, for any $t \in \textsf{Pre}(s)$, we get $D(\Theta_{\Delta}(t))=0$. As a result, condition C1 does not hold for diagnoser $D$, which contradicts the hypothesis.

($\Leftarrow$) Suppose that the system $G$ is OF-diagnosable. We consider a diagnoser $D:M(\mathcal{L}(G)) \to \{0,1\}$ defined by: for any $\omega \in \mathcal{L}_e(G)$
\begin{equation} \label{eq:dia}
	\begin{split}
		&D(\Theta_{\Delta}(\omega))= \\
		&\left\{
		\begin{array}{lr}
			1  &\text{if  } \forall s \in \mathcal{L}_e(G): \Theta_{\Delta}(\omega) = \Theta_{\Delta}(s) \Rightarrow \Sigma_{e,F} \in \Theta_{\Sigma}(s)  \\
			0  &otherwise.
		\end{array}
		\right. \\  
	\end{split}
\end{equation}
We claim that the diagnoser given by equation \eqref{eq:dia} satisfies condition C1 and C2.
We first show that the diagnoser satisfies condition C2. For any normal extended language $s \in \mathcal{L}_e(G) \backslash \mathcal{L}_{e,F}(G)$, by equation \eqref{eq:dia}, we have $D(\Theta_{\Delta}(s))=0$, that is, C2 holds.
To see that C1 holds, we consider any faulty extended string $s \in \mathcal{L}^{\varphi}_{e,F}(G)$. 
By OF-diagnosability, there exists a finite prefix $t \in \textsf{Pre}(s)$ such that any finite normal extended language $\omega$ having the same observation with $t$ is faulty, i.e., $\Sigma_{e,F} \in \Theta_{\Sigma}(\omega)$. Then, by equation \eqref{eq:dia}, we have $D(\Theta_{\Delta}(\omega))=1$, i.e., condition C1 also holds.
\end{proof}

We use  the following examples to illustrate the concept of output fairness as well as the notion of OF-diagnosability. 
 
\begin{myexm} \label{exm:3}
Again, let us consider system $G_1$ shown in Figure~\ref{fig:G-1} with $\Sigma_F=\{\sigma_{\!f}\}$ and suppose that  the fair outputs are   $\Delta'=\{b,c\}$. 
Then by Definition \ref{def:OF}, the $\Delta'$-fairness assumption is 
\[
	\varphi_{fair}
=
(   \bigwedge_{i=3,6}  
\Box \lozenge  b_i  \to  \Box \lozenge b^{b}_i
)
\wedge 
\left(\Box\lozenge  c_4  \to  \Box \lozenge c^{c}_4\right)
\]
Note that for infinite faulty  extended string 
\[
s_F=(1,\sigma_{\!f}, \varepsilon)[(2,a,a)(3,b,b)(4,c,\varepsilon)]^\omega,
\]
we have $s_F\not\models 	\varphi_{fair}$ 
because $\Box \lozenge  c_4 $ holds but $\Box \lozenge c^{c}_4$ does not hold. 
Therefore,  $s_F \not\in \mathcal{L}^{\varphi}_{e,F}(G)$ is not considered in the analysis of OF-diagnosability. 
Clearly, for any faulty string in $\mathcal{L}^{\varphi}_{e,F}(G)$, 
$c^c_4$ must occur, which means output $c$ will be observed and upon the occurrence of which the fault can be detected.
Therefore, this system is actually OF-diagnosable.   
\end{myexm}

\section{Verification of OF-diagnosability}\label{sec:4}
In this section, we investigate the verification of OF-diagnosability. 
Specifically, we provide a verifiable necessary and sufficient condition for OF-diagnosability based on the verification structure.

\subsection{Augmented System} 
To verify OF-diagnosability, our first step is to augment both the state-space and the event-space of $G$ such that 
\begin{itemize}
    \item 
    the information of whether or not a fault event has occurred is encoded in the augmented state-space; and 
    \item 
    the information of where the event occurs and  which specific output is observed are encoded in the augmented event-space. 
\end{itemize}

Formally, given  system $G = (Q,q_0,\Sigma,f)$, fault events $\Sigma_F$ and output function $\mathcal{O}$, we define the \emph{augmented system} as a new DFA 
\begin{equation} \label{eq:DFA}
	\tilde{G}=(\tilde{Q},\tilde{q}_0,\Sigma_e,\tilde{f}),
\end{equation}
where 
\begin{itemize}
	\item 
	$\tilde{Q} \subseteq Q \times \{F,N\}$ is the set of augmented states;
	\item 
	$\tilde{q}_0=(q_0,N)$ is the initial augmented state;
	\item 
	$\Sigma_e$ is the set of augmented events, which are just extended events;
	\item 
	$\tilde{f}:\tilde{Q} \times \Sigma_e \rightarrow \tilde{Q}$ is the transition function defined by: 
	for any $\tilde{q}=(q,l)\in \tilde{Q}$ and $\tilde{\sigma}=(q,\sigma,\delta) \in \Sigma_e$, we  have 
	$\tilde{f}(\tilde{q},\tilde{\sigma})!$ if 
	$f(q,\sigma)!$ and $\delta\in \mathcal{O}(q,\sigma)$. 
	Furthermore, when $\tilde{f}(\tilde{q},\tilde{\sigma})!$, we have 
	\[
	\tilde{f}(\tilde{q},\tilde{\sigma})=
		\left\{ 
		\begin{array}{l l}
		(f(q,\sigma),N) &\text{ if } l=N\wedge \tilde{\sigma} \!\notin\! \Sigma_{e,F} \\ 
		(f(q,\sigma),F) &\text{ otherwise }
		\end{array} 
		\right..  
	\] 
\end{itemize}

The above constructed augmented system $\tilde{G}$ has the following properties:
\begin{itemize}
    \item 
    First, the augmented system $\tilde{G}$ generates extended strings. 
    Essentially, it still tracks the original dynamic of the system by putting both the output realization and the current state information together with the internal event. 
    Therefore, we have 
    \[
    \mathcal{L}(\tilde{G})=\mathcal{L}_e(G)
    \text{ and }
    \mathcal{L}^\omega(\tilde{G})=\mathcal{L}_e^\omega(G).
    \]
    \item 
    Second, each augmented state $(q,l) \in Q \times \{F,N\}=\tilde{Q}$ has two components. 
   The first component $q$ is the actual state in the original system $G$ and the second component $l \in \{N,F\}$ is a label denoting whether fault events have occurred. 
   By the construction, 
  the label will change from $N$ to $F$ only when an extended fault event occurs and 
  once the label becomes $F$, it will be $F$ forever. 
   We denote by $\tilde{Q}_N=\{(q,l)\in \tilde{Q}: l = N\}$ the set of normal augmented states and the set of faulty states  $\tilde{Q}_F$ is defined analogously.  
\end{itemize}

\begin{figure} 
	\centering
	\centering
	\includegraphics[width=3cm]{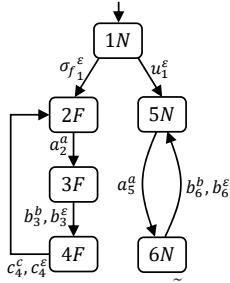}
	\caption{\label{fig:G1-1}Augmented system $\tilde{G}_1$ for system $G_1$.}
\end{figure}
\begin{myexm}
Still, we consider system $G_1$ shown in Figure~\ref{fig:G-1} with the same setting in Example~\ref{exm:3}. 
Its augmented system $\tilde{G}_1$   is depicted in Figure~\ref{fig:G1-1}, where $(1,\sigma_{\!f}, \varepsilon)$ is the extended fault event and after occurrence of which all states is augmented with label $F$, e.g., states $2F$ and $3F$. 
Furthermore, the transitions of the augmented system $\tilde{G}_1$ are defined according to the actual transitions and the underlying observations of original system $G_1$. For example, because $f(3,b)=4$ and $\mathcal{O}(3,b)=\{b,\varepsilon\}$, we have two new transitions in $\tilde{G}_1$: $\tilde{f}(3F, b_3^b)=4F$ and $\tilde{f}(3F, b_3^\varepsilon)=4F$.
\end{myexm}


\subsection{Verification Structure}	
 
Let $r \in 2^{\tilde{Q}}$ be a set of states representing the current state estimation of the augmented system. 
By observing a new output symbol $\delta\in \Delta$, the estimate is updated by the following observable reach
\[
\textsf{Next}(r,\delta)=
\left\{  \tilde{q}' \in \tilde{Q}
: \!\!\!
\begin{array}{cc}
\exists \tilde{q}\in r, \tilde{\sigma}\in \Sigma_e \text{ s.t } \\
\Theta_{\Delta}(\tilde{\sigma})\!=\!\delta  \wedge \tilde{q}'\!=\!\tilde{f}(\tilde{q},\tilde{\sigma})
\end{array}
\right\}.
\]
Also, the system can execute silently without generating an output symbol. 
This is captured by the unobservable reach  defined by: 
\[
\textsf{UR}(r)=
\left\{
 \tilde{q}' \in \tilde{Q}
: \!\!\!
\begin{array}{cc}
\exists \tilde{q}\in r, s\in \Sigma_e^* \text{ s.t } \\
\Theta_{\Delta}(s)\!=\!\varepsilon  \wedge \tilde{q}'\!=\!\tilde{f}(\tilde{q},s)
\end{array}
\right\}.
\]

Now, based on the augmented system $\tilde{G}= (\tilde{Q},\tilde{q}_0,\tilde{\Sigma},\tilde{f})$, 
we construct the  \emph{verification structure} as follows:
\begin{equation}
	V=(Q_V,q_{V}^0,\Sigma_V,f_V)
\end{equation}
where
\begin{itemize}
	\item 
	$Q_V \subseteq \tilde{Q} \times 2^{\tilde{Q}}$ is the  set of states;
	\item 
	$q_{V}^0=(\tilde{q}_0,\textsf{UR}(\{\tilde{q}_0\}))$ is the initial state;
	\item 
	$\Sigma_V=\Sigma_e$ is the event set, which is just the set of extended events;
	\item 
	$f_V:Q_V \times \Sigma_V \rightarrow Q_V$ is the transition function defined by: 
	for any $(\tilde{q},r) \in \tilde{Q} \times 2^{\tilde{Q}}$ and $\tilde{\sigma}\in \Sigma_V$, we have
	\[
	f_V( (\tilde{q},r), \tilde{\sigma} )=
	\left\{\!\! 
	\begin{array}{l l}
		(\tilde{q}',r) &\text{if }   \Theta_{\Delta}(\tilde{\sigma})=\varepsilon \\ 
		(\tilde{q}',r')  &\text{if } \Theta_{\Delta}(\tilde{\sigma})\in \Delta
	\end{array} 
	\right.\]  
	where $\tilde{q}'=\tilde{f}(\tilde{q},\tilde{\sigma})$
	and $r'=\textsf{UR}( \textsf{Next}(r,  \Theta_{\Delta}(\tilde{\sigma})  ) )$.
\end{itemize}

Intuitively, each state $(\tilde{q},r)$ in the verification structure $V$ has two components. 
The first component $\tilde{q} \in \tilde{Q}$ tracks the current (augmented) state in $\tilde{G}$. Therefore, the transition of the first part is consistent with $\tilde{f}$. 
As a result, we also have 
\[
\mathcal{L}(V)=\mathcal{L}(\tilde{G})=\mathcal{L}_e(G)
\]
and the same for the generated infinite language. 
On the other hand, the second component  $r \in 2^{\tilde{Q}}$ captures the current state estimation of system $\tilde{G}$. 
Therefore, this component is updated only when a new output symbol is generated, i.e.,  
$\Theta_{\Delta}(\tilde{\sigma})\in \Delta$. 
Furthermore, if it is updated, then it first updates the estimate by the observable reach to compute the set of state that can be reached immediately following the output. 
Also, we need to include all states that can be reached silently by using the unobservable reach. 
Essentially, we can image $V$ as the synchronization of the augmented system $\tilde{G}$ and its current-state estimate under the non-deterministic observation setting.

For any state $(\tilde{q},r) \in Q_V$ in $V$, we say the state is
\begin{itemize}
	\item 
	\emph{certain} 
	if $\tilde{q}\in \tilde{Q}_F$ and $r\subseteq \tilde{Q}_F$;
	\item 
	\emph{uncertain} 
	if $\tilde{q}\in \tilde{Q}_F$ and $r \cap \tilde{Q}_N\neq \emptyset$.
\end{itemize}
We denote by $Q_V^{cer}$ and $Q_V^{unc}$ the set of certain states and uncertain states in $V$, respectively. 
Then for any string $s \in \mathcal{L}(V)$, 
based on the definition of uncertain and certain states, we have following properties: 
\begin{itemize}
	\item
	For any faulty  extended string  $s \in \mathcal{L}(V)=\mathcal{L}_e(G)$, 
there exists a normal extended string $s' \in \mathcal{L}_e(G)$ such that $\Theta_{\Delta}(s)=\Theta_{\Delta}(s')$ if and only if $f_V(q_V^0,s) \in Q_V^{unc}$.
	\item 
	For any faulty  extended  string  $s \in \mathcal{L}(V)=\mathcal{L}_e(G)$,
	if $f_V(q_V^0,s) \in Q_V^{cer}$, 
	then for any of its continuation $s t \in \mathcal{L}(V)$, we have $f_V(q_V^0, s t) \in Q_V^{cer}$.
\end{itemize}

\subsection{Checking OF-Diagnosability}
Now we investigate the verification of OF-diagnosability. 
According to Definition~\ref{def:OF}, a system is not OF-diagnosable if and only if 
there exists an infinite extended faulty string satisfying the $\Delta'$-fairness assumption but all states in $V$ reached along the string are uncertain. 
Then, since the systems is finite, only loops can generate infinite strings. This leads to the definition of \emph{fairly uncertain loop} (FU-loop). 

Formally, given the verification structure $V$, we define a \emph{run} in $V$ as a finite sequence 
\[
\pi=q_V^1 \xrightarrow{\sigma_V^1} q_V^2 \xrightarrow{\sigma_V^2} \cdots \xrightarrow{\sigma_V^{n-1}} q_V^n 
\]
 where $q_V^1,\cdots,q_V^n \in Q_V, \sigma_V^1,\cdots,\sigma_V^n \in \Sigma_V$ and $q_V^{i+1} = f_V(q_V^i,\sigma_V^i), i=1,\dots,n-1$. 
A run of the above form is called a \emph{loop} if $q_V^1=q_V^n$. 

Then given a  loop $\pi = q_V^1 \xrightarrow{\sigma_V^1} q_V^2 \cdots  \xrightarrow{\sigma_V^{n-1}} q_V^n$ in the verification structure, 
where $\sigma_V^i=( q^i,\sigma^i,\delta^i )$, 
we say $\pi$ is 
\begin{itemize}
    \item 
    \emph{fair} if for each transition that occurs in the loop, any of its fair output symbols must occur in the loop associating with the same transition, i.e.,  
		\begin{align}
		&(\forall i \in \{1, \cdots, n\})
		(\forall \delta\in \mathcal{O}(q^i,\sigma^i)\cap \Delta')\nonumber\\
		&
		(\exists j \in \{1, \cdots, n\})
		[(q^j,\sigma^j)\!=\!(q^i,\sigma^i) \wedge \delta^j\!=\!\delta]
		\nonumber
		\end{align}
	\item 
	\emph{uncertainty} if all states in the loop are uncertain, i.e., 
		\[
		\forall i \in \{1, \cdots, n\}: q_V^i \in Q_V^{unc}
		\]
	\item 
	\emph{reachable} if  there exists a finite string $s \in \mathcal{L}(V)$ such that $f_V(q_V^0,s)=q_V^1$. 
\end{itemize}
Clearly, by properties of the verification structure, we know that 
an uncertain loop will only be reached by faulty strings. 
Furthermore, if some state in a loop is uncertain, then all states in it are uncertain. 

In the following two lemmas, we formally establish the relationship between strings satisfying the $\Delta'$-fairness assumption and fair loops, according to the structure characteristics of $\Delta'$-fairness assumption. Due to space constraint, their proofs are omitted here and are available in \cite{sup-material}.

First, we show that, for any reachable fair loop, it can generate an infinite extended string string satisfying the $\Delta'$-fairness assumption. 
\begin{mylem} \label{lem:satisfy OF}
Given a  reachable fair loop $\pi = q_V^1 \xrightarrow{\sigma_V^1} q_V^2 \cdots  \xrightarrow{\sigma_V^{n-1}} q_V^n$ in verification system $V$, 
there exists an infinite extended string  in $V$ in the form of
\[
s=t\left(  \sigma_V^1\sigma_V^2 \cdots \sigma_V^n  \right)^\omega\in \mathcal{L}^\omega(V)=\mathcal{L}^\omega_e(G)
\]
such that $s\models \varphi_{fair}$.
\end{mylem} 
\begin{proof}
By contradiction, we assume $s=t\left(  \sigma_V^1\sigma_V^2 \cdots \sigma_V^n  \right)^\omega$ does not satisfy $\varphi_{fair}$, that is, there exists $\sigma \in \Sigma, q \in Q$ such that $s$ satisfies LTL formula $\Box \diamondsuit \sigma_q$, while there exists $\delta \in \Delta' \cap \mathcal{O}(q,\sigma)$ and $s$ does not satisfy formula $\Box \diamondsuit \sigma^\delta_q$. 
By the first formula $\Box \diamondsuit \sigma_q$, we have there is $\delta' \in \mathcal{O}(q,\sigma), \delta' \neq \delta$ such that $\sigma^{\delta'}_q $ appears in $s$ for infinite times, that is, $\sigma^{\delta'}_q $ is concluded in $ \sigma_V^1\sigma_V^2 \cdots \sigma_V^n$, and since the formula $\Box \diamondsuit \sigma^\delta_q$ is false, $\sigma^{\delta}_q$ cannot occur for infinite times, that is, $\sigma^{\delta}_q$ is not contained by $ \sigma_V^1\sigma_V^2 \cdots \sigma_V^n$. 
Thus, the loop $\pi$ is not fair and the hypothesis is contradicted.	
\end{proof}
 
On the other hand, 
given an infinite extended string  satisfying $\varphi_{fair}$, 
it will also induce a fair loop. 

\begin{mylem} \label{lem:loop existence}
For any infinite extended string $s \in \mathcal{L}^\omega(V)=\mathcal{L}^\omega_e(G)$ satisfying  $\varphi_{fair}$, 
we can find a reachable fair loop $\pi = q_V^1 \xrightarrow{\sigma_V^1} q_V^2 \xrightarrow{\sigma_V^2} \cdots \xrightarrow{\sigma_V^{n-1}} q_V^n$ in which events in the loop are the same as those event appearing infinite number of times in $s$, i.e., $\{ \sigma_V^1,\dots,\sigma_V^n \} = \textsf{Inf}(s)$.
\end{mylem} 
\begin{proof}
	Suppose there is an infinite string $s \in \mathcal{L}(V)$ and $s$ satisfies $\varphi_{fair}$. 
	Because system $V$ is finite, we can obtain a reachable loop with the structure $\pi = q_V^1 \xrightarrow{\sigma_V^1} q_V^2 \xrightarrow{\sigma_V^2} \cdots \xrightarrow{\sigma_V^{n-1}} q_V^n$ such that the set of events in loop $\pi$ is equal to $\textsf{Inf}(s)$, i.e., $\{\sigma_V\in \Sigma_V: \exists i\in \{1,\cdots,n\},\sigma_V=\sigma_V^i\} = \textsf{Inf}(s)$.
	We assume the loop $\pi$ is not fair, by contradiction. 
	That is, there exists $ i \in \{1,\cdots, n\}$ and a fair output $\delta$ in the set $\mathcal{O}(q^i, \sigma^i) \cap \Delta'$ such that for all $j \in \{1,\cdots, n\}$, we have either $(q^j,\sigma^j) \neq (q^i,\sigma^i)$ or $\delta^j \neq \delta$.
	Let $q=q^i=q^j, \sigma=\sigma^i=\sigma^j$ and $\delta'=\delta^i$. As a result, we know that $\sigma^{\delta'}_q $ is included by $ \textsf{Inf}(s)$ but $\sigma^{\delta}_q $ is not. The formula $\Box \lozenge \sigma_q$ is true, while the formula $\Box \lozenge \sigma^{\delta}_q$ is false. Then, the hypothesis is contradicted.	
\end{proof}

Based on  Lemmas \ref{lem:satisfy OF} and \ref{lem:loop existence}, 
we obtain the following main result, which shows that, 
to check OF-diagnosability, it suffices to check the existence of a reachable fair and uncertain loop (FU-loop) in the verification structure $V$. 

\begin{mythm}\label{thm:verify OF-diag}
	A system $G$ is not OF-diagnosable w.r.t. fault events $\Sigma_F$, output function $\mathcal{O}$ and fair outputs $\Delta' \subseteq \Delta$, if and only if, there exists a reachable FU-loop in the verification structure $V$. 
\end{mythm}
\begin{proof}
($\Leftarrow$) Assume system $G$ is OF-diagnosable. 
That is, for any infinite faulty extended string $s \in \mathcal{L}^\varphi_{e,F}(G)$, it has a finite prefix $t$ such that each extended string $\omega \in \mathcal{L}_e(G)$, which has the same output with $t$, is faulty, i.e., $\Sigma_{e,F} \in \omega$. 

By contradiction, we claim there exists a reachable  FU-loop $\pi$ in $V$. 
Suppose $\pi=q_V^1 \xrightarrow{\sigma_V^1} q_V^2 \xrightarrow{\sigma_V^2} \cdots \xrightarrow{\sigma_V^{n-1}} q_V^n$. 
By Lemma 1, there is an infinite string $s=t\left(  \sigma_V^1\sigma_V^2 \cdots \sigma_V^n  \right)^\omega$ in $V$ satisfying $\varphi_{fair}$, where $t \in \textsf{Pre}(s)$ and $f_V(q_V^0,s)=q^1_V$. 
By properties of verification structure, we know that for any prefix $t' \in \textsf{Pre}(s)$, $f(q^0_V,t') \notin Q^{cer}_V$. 
And since all states in loop $\pi$ are uncertain, for any prefix $t' \in \textsf{Pre}(s)$, there exists a normal extended string $\omega \in \mathcal{L}_e(G)$ having the same output with $t'$. The system is not OF-diagnosable, which contradicts the hypothesis.

($\Rightarrow$) Suppose system $G$ is not OF-diagnosable. 
By Definition 3, there is an infinite faulty extended string $s \in \mathcal{L}^\varphi_{e,F}(G)$ such that for any prefix $t \in \textsf{Pre}(s)$, a normal extended string $\omega \in \mathcal{L}_e(G)$ has the same output with $t$. 
Because of this normal extended string $\omega$, by  properties of verification structure, any prefix $t$ of the infinite faulty extended string $s$ cannot reach certain states, i.e., $f_V(q^0_V, t) \notin Q_V^{cer}$. 
Because extended string $s$ is fair, i.e., $s \models \varphi_{fair}$, by Lemma 2, there exists a fair loop $\pi = q_V^1 \xrightarrow{\sigma_V^1} q_V^2 \xrightarrow{\sigma_V^2} \cdots \xrightarrow{\sigma_V^{n-1}} q_V^n$, where $\{ \sigma_V^1,\dots,\sigma_V^n \} = \textsf{Inf}(s)$. 
For any faulty prefix $t' \in \textsf{Pre}(s)$, there is a string $t'' \in \Sigma_e^*$ such that $t' t''$ is a prefix of $s$ and reaches the state $q_V^1$, i.e., $f_V(q_V^0,t' t'')=q_V^1$. By properties of verification structure, all states in loop $\pi$ are uncertain, that is, $\pi$ is a reachable FU-loop. The proof is complete.
\end{proof}
 
The condition in Theorem \ref{thm:verify OF-diag} can be checked as follows. 
First, in verification structure, we find all strongly connected components (SCCs) that are reachable by fault events. 
Clearly, each of above SCCs only consists of either certain states or uncertain states and we only need to consider those uncertain SCCs.
Then we need to check if any of the uncertain SCCs contains a fair loop. 
To this end, for each extended event $(q,\sigma,\delta)$, we check if it is fair in the sense that  any $\delta\in \Delta' \cap \mathcal{O}(q,\sigma)$ also appears in the same SCC. 
If not, we need to remove such an extended event, and then recompute the SCCs and repeat the above removal procedure. 
When no such ``unfair" extended event can be removed, the SCC remained (if exists) will contain a fair-loop; otherwise, no fair-lopp can be found.  
The above procedure is in polynomial-time in the size of the verification structure since computing all SCCs can be done in linear time and the above procedure repeats at most $|Q|^2|\Sigma||\Delta|$ times. 
However, the size of the verification structure is exponential in the size of the original system. 
Therefore, the overall complexity for checking OF-diagnosability is exponential.

\begin{myexm}
Still,  let us consider system $G_1$ shown  Figure~\ref{fig:G-1}
with $\Sigma_F=\{\sigma_{\!f}\}$ and suppose that  the fair outputs are   $\Delta'=\{b,c\}$. 
As we have discussed in Example \ref{exm:3}, this system is  OF-diagnosability. Here, we analyze this more formally using Theorem \ref{thm:verify OF-diag}. 
Based on the augmented system $\tilde{G}_1$, we construct the verification structure $V_1$; 
part of it  is shown in Figure~\ref{fig:G-1-3}, where we focus on the only reachable uncertain loop, as highlighted with red lines in the Figure~\ref{fig:G-1-3}, and omit other parts without loss of generality for the purpose of verification. 
To form a loop, we have to fire extended event $c_4^\epsilon$ infinite number of times  at state $4F$. 
However,  the fair output of transition $4\xrightarrow{c}$ is $c\in \Delta'$, i.e., $\mathcal{O}(4,c)\cap \Delta'=\{c\}$ 
and extended  $c_4^c$ is not included. 
Therefore, we conclude the $V_1$ cannot form a FU-loop, which means that
system $G_1$ is OF-diagnosable according to Theorem \ref{thm:verify OF-diag}. 
\end{myexm}
\begin{figure}[htbp]
	\centering
	\centering
	\includegraphics[width=6.5cm]{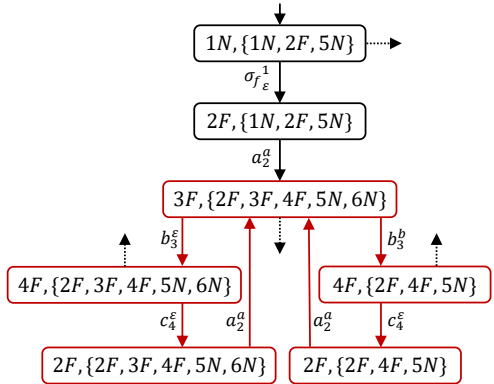}
	\caption{\label{fig:G-1-3}Part of verification structure $V_1$ corresponding to system $G_1$. The uncertain loop is highlighted by red lines.}
\end{figure}

In the previous running example, the non-deterministic observation is  either to see the occurrence of the internal event or to lose the observation. This actually corresponds to the special case of intermittent loss of observations. 
As we mentioned, our framework captures  the general case of state-dependency and allows the
output symbols to be different from the event set.  
In the following example, we consider such a general scenario. 

\begin{myexm}
Let us consider  system $G_2$ shown in Figure~\ref{fig:G-2} with $\Sigma=\{a,b,c,\sigma_{\!f},u\}$, $\Delta=\{o_1,o_2,o_3\}$ and $\Sigma_F=\{\sigma_{\!f}\}$. 
We assume that the fair outputs are $\Delta'=\{o_2,o_3\}$. 
For example, for transition $4\xrightarrow{c}$, we may observe any output $\delta \in \mathcal{O}(4,c)=\{o_2,o_3,\varepsilon\}$ non-deterministically. 
The  augmented system $\tilde{G}_2$ is  depicted in Figure~\ref{fig:G-2-2}. 
For this system, the  $\Delta'$-fairness assumption is given by 
\begin{align}
    \varphi_{fair}=
    &(\bigwedge_{i=4,7} (\Box \lozenge  c_i  \to \bigwedge_{\delta \in \{o_2,o_3\}} \Box \lozenge c^{\delta}_i))\nonumber \\
    &\wedge (\Box \lozenge  b_3  \to \bigwedge_{\delta \in \{o_2,o_3\}} \Box \lozenge b^{\delta}_3) \wedge \left(\Box\lozenge  b_6  \to  \Box \lozenge b^{o_2}_6\right). \nonumber \vspace{-12pt}
\end{align} 
Now, let us verify OF-diagnosability using  Theorem \ref{thm:verify OF-diag}. 
To this end, we construct the corresponding verification system $V_2$, part of which is shown in  Figure~\ref{fig:G-2-3}. 
However, there is fair but uncertainty loop reachable, which is highlighted in red color. 
Specifically, we consider the following loop
\begin{align}
	&(3F,\{2F,3F,4F,7N\})
	\xrightarrow{b_3^\varepsilon}
	(4F,\{2F,3F,4F,7N\})\nonumber\\
	\xrightarrow{c_4^{o_2}}
	&(2F,\{2F,4F,8N\})  
	\xrightarrow{a_2^{o_1}} 
	(3F,\{2F,3F,4F,7N\})\nonumber\\
	\xrightarrow{b_3^{o_2}}
	&(4F,\{2F,4F,8N\})  
	\xrightarrow{c_4^{\varepsilon}} 
	(2F,\{2F,4F,8N\}) \nonumber\\ 
	\xrightarrow{a_2^{o_1}} 
	&(3F,\{2F,3F,4F,7N\})
	\xrightarrow{b_3^\varepsilon}
	(4F,\{2F,3F,4F,7N\})\nonumber\\
	\xrightarrow{c_4^{o_3}}
	&(2F,\{2F,4F,8N\})  
	\xrightarrow{a_2^{o_1}} 
	(3F,\{2F,3F,4F,7N\})\nonumber\\
	\xrightarrow{b_3^{o_3}}
	&(4F,\{2F,4F,8N\})  
	\xrightarrow{c_4^{\varepsilon}} 
	(2F,\{2F,4F,8N\}) \nonumber\\ 
	\xrightarrow{a_2^{o_1}} 
	&(3F,\{2F,3F,4F,7N\}).\nonumber 
\end{align}
This loop is fair because for transition $3\xrightarrow{b}$ all fair outputs
in $\mathcal{O}(3,b) \cap \Delta'=\{o_2,o_3\}$ occur in the loop, and for transition $4\xrightarrow{c}$ all fair outputs
in $\mathcal{O}(4,c) \cap \Delta'=\{o_2,o_3\}$ occur in the loop. 
Furthermore, all states in the loop are uncertain. 
Thus, system $G_2$ is not OF-diagnosable.
\end{myexm}

\begin{figure}[htbp] 
	\centering
	\subfigure[System $G_2$.]
		{\label{fig:G-2}
		\centering
		\includegraphics[width=0.4\linewidth]{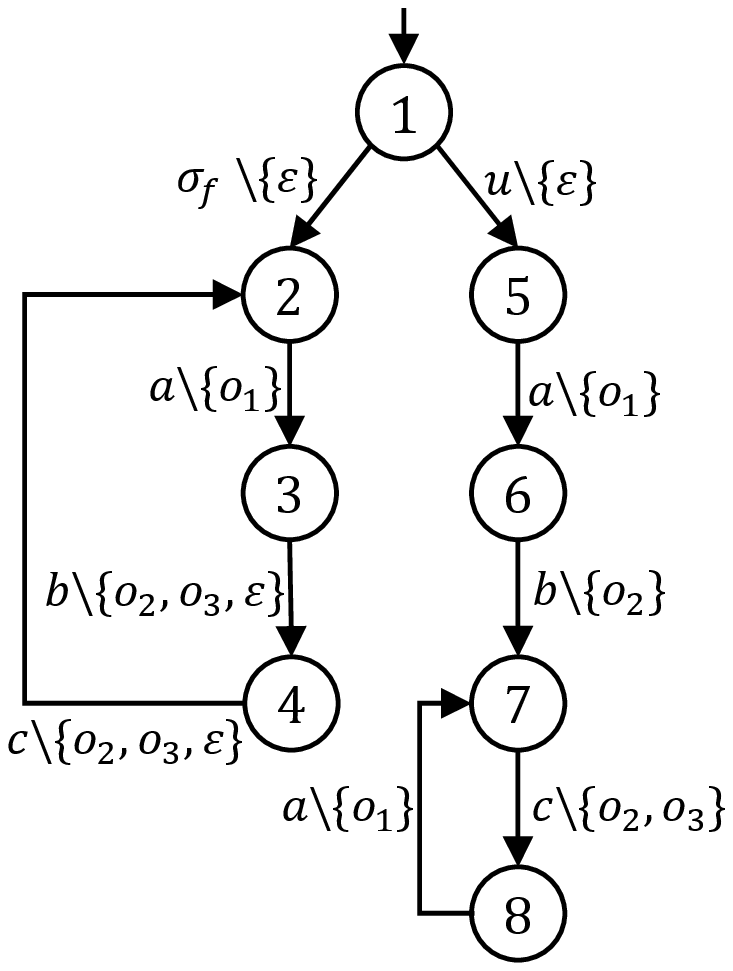}
		}
	\quad
	\subfigure[Augmented system $\tilde{G}_2$.]
	{\label{fig:G-2-2}
	\centering
	\includegraphics[width=0.4\linewidth]{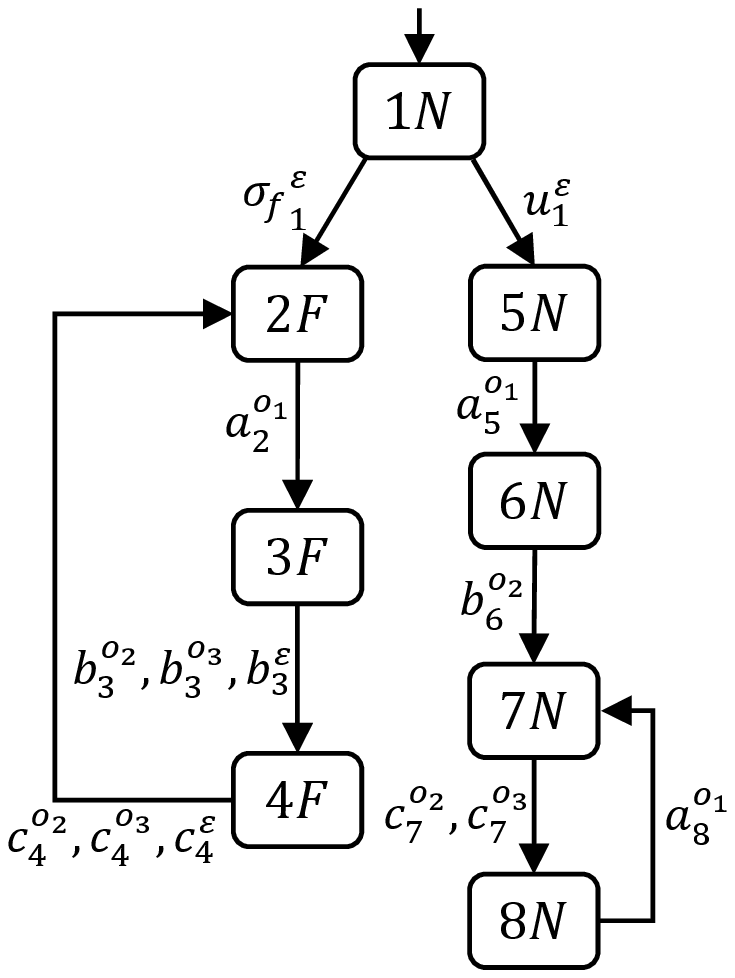}}
	\caption{Example of a non-OF-diagnosable system.} \label{fig:GG}
\end{figure}

\begin{figure}[htbp]
	\centering
	\centering
	\includegraphics[width=5.5cm]{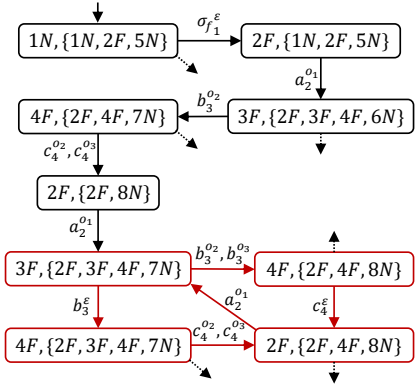}
	\caption{\label{fig:G-2-3} Part of verification structure $V_2$ corresponding to system $G_2$. The uncertain loop is highlighted by red lines.}
\end{figure}

\section{Conclusion}\label{sec:5}
In this work, we revisited the problem of fault diagnosis of DES under non-deterministic observations. 
Compared with existing works, we introduced the notion of output fairness that excludes the case where some possible outputs are disabled permanently, which is formalized as $\Delta'$-fairness assumption by LTL. 
We proposed a new notion called OF-diagnosability to capture the diagnostic requirement under $\Delta'$-fairness assumption.
Necessary and sufficient condition for testing OF-diagnosability was also provided based. 
Our work bridged the discrepancy between the existing definition of diagnosability under non-deterministic observations and the physical setting of non-deterministic observations. 

\bibliographystyle{plain}
	\bibliography{des}

\begin{thebibliography}{10}

\bibitem{baier2008principles}
C.~Baier and J.~Katoen.
\newblock {\em Principles of Model Checking}.
\newblock MIT press, 2008.

\bibitem{biswal2015polynomial}
P.~Biswal and S.~Biswas.
\newblock A polynomial algorithm for diagnosability of fair discrete event
  systems.
\newblock {\em Systems Science \& Control Engineering}, 3(1):307--319, 2015.

\bibitem{boussif2021intermittent}
A.~Boussif, M.~Ghazel, and J.~Basilio.
\newblock Intermittent fault diagnosability of discrete event systems: an
  overview of automaton-based approaches.
\newblock {\em Discrete Event Dynamic Systems}, 31(1):59--102, 2021.

\bibitem{carvalho2012robust}
L.~Carvalho, J.~Basilio, and M.~Moreira.
\newblock Robust diagnosis of discrete event systems against intermittent loss
  of observations.
\newblock {\em Automatica}, 48(9):2068--2078, 2012.

\bibitem{carvalho2021comparative}
L.~Carvalho, M.~Moreira, and J.~Basilio.
\newblock Comparative analysis of related notions of robust diagnosability of
  discrete-event systems.
\newblock {\em Annual Reviews in Control}, 2021.

\bibitem{carvalho2013robust}
L.~Carvalho, M.~Moreira, J.~Basilio, and S.~Lafortune.
\newblock Robust diagnosis of discrete-event systems against permanent loss of
  observations.
\newblock {\em Automatica}, 49(1):223--231, 2013.

\bibitem{cassandras2008introduction}
C.~Cassandras and S.~Lafortune.
\newblock {\em Introduction to Discrete Event Systems}, volume~2.
\newblock Springer, 2008.

\bibitem{hu2021diagnosability}
Y.~Hu, Z.~Ma, Z.~Li, and A.~Giua.
\newblock Diagnosability enforcement in labeled petri nets using supervisory
  control.
\newblock {\em Automatica}, 131:109776, 2021.

\bibitem{kanagawa2015diagnosability}
N.~Kanagawa and S.~Takai.
\newblock Diagnosability of discrete event systems subject to permanent sensor
  failures.
\newblock {\em International Journal of Control}, 88(12):2598--2610, 2015.

\bibitem{lafortune2018history}
S.~Lafortune, F.~Lin, and C.~Hadjicostis.
\newblock On the history of diagnosability and opacity in discrete event
  systems.
\newblock {\em Annual Reviews in Control}, 45:257--266, 2018.

\bibitem{lin2014control}
F.~Lin.
\newblock Control of networked discrete event systems: dealing with
  communication delays and losses.
\newblock {\em SIAM Journal on Control and Optimization}, 52(2):1276--1298,
  2014.

\bibitem{lin2017n}
F.~Lin, W.~Chen, L.~Han, B.~Shen, et~al.
\newblock N-diagnosability for active on-line diagnosis in discrete event
  systems.
\newblock {\em Automatica}, 83:220--225, 2017.

\bibitem{oliveira2020k}
V.~Oliveira, F.~Cabral, and M.~Moreira.
\newblock K-loss robust diagnosability of discrete-event systems.
\newblock {\em IFAC-PapersOnLine}, 53(4):250--255, 2020.

\bibitem{ran2018codiagnosability}
N.~Ran, H.~Su, A.~Giua, and C.~Seatzu.
\newblock Codiagnosability analysis of bounded petri nets.
\newblock {\em IEEE Trans.\ Automatic Control}, 63(4):1192--1199, 2018.

\bibitem{sampath1995diagnosability}
M.~Sampath, R.~Sengupta, S.~Lafortune, K.~Sinnamohideen, and D.~Teneketzis.
\newblock Diagnosability of discrete-event systems.
\newblock {\em IEEE Trans.\ Automatic Control}, 40(9):1555--1575, 1995.

\bibitem{takai2021general}
S.~Takai.
\newblock A general framework for diagnosis of discrete event systems subject
  to sensor failures.
\newblock {\em Automatica}, 129:109669, 2021.

\bibitem{takai2012verification}
S.~Takai and T.~Ushio.
\newblock Verification of codiagnosability for discrete event systems modeled
  by {M}ealy automata with nondeterministic output functions.
\newblock {\em IEEE Trans.\ Aut.\ Control}, 57(3):798--804, 2012.

\bibitem{ushio2015nonblocking}
T.~Ushio and S.~Takai.
\newblock Nonblocking supervisory control of discrete event systems modeled by
  {M}ealy automata with nondeterministic output functions.
\newblock {\em IEEE Trans.\ Aut.\ Control}, 61(3):799--804, 2015.

\bibitem{viana2019codiagnosability}
G.~Viana, M.~Moreira, and J.~Basilio.
\newblock Codiagnosability analysis of discrete-event systems modeled by
  weighted automata.
\newblock {\em IEEE Trans.\ Automatic Control}, 64(10):4361--4368, 2019.

\bibitem{yin2019robust}
X.~Yin, J.~Chen, Z.~Li, and S.~Li.
\newblock Robust fault diagnosis of stochastic discrete event systems.
\newblock {\em IEEE Trans.\ Automatic Control}, 64(10):4237--4244, 2019.

\bibitem{yin2017decidability}
X.~Yin and S.~Lafortune.
\newblock On the decidability and complexity of diagnosability for labeled
  {P}etri nets.
\newblock {\em IEEE Trans.\ Automatic Control}, 62(11):5931--5938, 2017.

\bibitem{zaytoon2013overview}
J.~Zaytoon and S.~Lafortune.
\newblock Overview of fault diagnosis methods for discrete event systems.
\newblock {\em Annual Rev.\ Control}, 37(2):308--320, 2013.

\end{thebibliography}

\end{document}